\documentclass[pra,twoside,amsmath,amsfonts,amssymb,superscriptaddress,showpacs]{revtex4}
\usepackage{amsmath}
\usepackage{amscd}
\usepackage{float}
\usepackage{longtable}
\usepackage{epsfig}
\usepackage{amssymb}
\usepackage{amsfonts}
\usepackage{latexsym}
\usepackage{theorem}
\usepackage{times}
\usepackage{bm}
\usepackage[pdftex]{color}
\usepackage[cmtip,arrow]{xy}
\usepackage{pb-diagram,pb-xy}
\usepackage{verbatim}
\usepackage{fancyhdr}
\usepackage{tikz}
\usepackage{amscd}
\usepackage{times}
\usepackage{titlesec}
\usepackage[english]{babel}
\usepackage{pgf,pgfarrows,pgfnodes,pgfautomata,pgfheaps}
\usepackage{graphicx}
\usepackage{longtable}
\usepackage{epsfig}
\newtheorem{thm}{Theorem}

\newtheorem{lem}[thm]{Lemma}
\newtheorem{prop}[thm]{Proposition}

\newtheorem{rem}[thm]{Remark}

\newenvironment{proof}[1][{}]{\par\vskip12ptplus30pt\noindent{\it Proof{#1}:\ }}%
{\hfill$\blacksquare$\par\vskip12ptplus60pt}
\def\1#1{{\bf #1}}
\def\2#1{{\mathcal #1}}
\def\4#1{{\tt #1}}
\def\5#1{{\sf #1}}
\def\6#1{{\mathfrak #1}}
\def\7#1{{\Bbb #1}}
\def\8#1{{\rm #1}}
\def\9#1{{\mathcurl #1}}

\DeclareFontFamily{OT1}{rsfs}{}
\DeclareFontShape{OT1}{rsfs}{m}{n}{<-7> rsfs5 <7-10> rsfs7 <10-> rsfs10}{}
\DeclareMathAlphabet\mathcurl{OT1}{rsfs}{m}{n}
\definecolor{grey}{rgb}{0.5,0.5,0.5}

\def\I{\openone}

\newcommand{\SPROD}[2]{\langle{#1},{#2}\rangle}

\tikzstyle{prep}=[draw=black!80, fill=red!40, thick, minimum width=0.5cm, text centered, minimum height=0.5cm]
\tikzstyle{meas}=[draw=black!80, fill=green!40, thick, minimum width=0.6cm,
text centered, minimum height=2cm]
\tikzstyle{locmeas}=[draw=black!80, fill=blue!40, thick, minimum width=0.5cm,
text centered, minimum height=0.5cm]
\tikzstyle{oper}=[draw=black!80, fill=blue!40, thick, minimum width=0.5cm, text centered, minimum height=0.5cm]
\tikzstyle{textbox}=[draw, fill=white, text width=3cm, text centered, minimum height=.6cm]
\tikzstyle{emptybox}=[text centered, minimum height=.6cm,anchor=south]
\tikzstyle{shadow}=[draw=black!50, fill=black!50, text width=3cm, text centered, minimum height=.6cm]
\tikzstyle{entangnb}=[shape = circle, draw=red, fill=red]
\tikzstyle{entangns}=[shape = circle, draw=red, fill=red, scale=.5]
\tikzstyle{redb}=[shape = circle, draw=red, fill=red, scale=.5]
\tikzstyle{blueb}=[shape = circle, draw=blue, fill=blue, scale=.5]
\tikzstyle{grayb}=[shape = circle, draw=gray, fill=gray, scale=.5]

\tikzstyle{quant}=[->,snake=snake,draw=blue,line width=2pt,line after snake=1mm] 
\tikzstyle{quant1}=[snake=snake,draw=blue,line width=2pt] 
\tikzstyle{class}=[line width=2pt]
\tikzstyle{entang}=[snake=coil,draw=red,line width=2pt,segment amplitude=8pt,line after snake=1mm] 
\tikzstyle{entang1}=[snake=coil,draw=red, segment amplitude=3pt]

\begin{document}
%%%%%%%%%%%%%%%%%%%%%%%%%%%%%%%%%%%%%%%%%%%%%%%%%%%%%%%%%%%%%%%%%%%%%%%%%%%%%%%%%%%%%%%%%%%%%%%%%%%%
\pagestyle{fancy}
\title{On the bosonic behavior of mean-field fluctuations in atomic ensembles}

\author{Michael Keyl}
\email{d.schlingemann@tu-bs.de}
\affiliation{ISI Foundation, Quantum Information Theory Unit,\\
Viale S. Severo 65, 10133 Torino, Italy}
\affiliation{Institut f\"ur Mathematische Physik, Technische
Universit\"at Braunschweig, Mendelssohnstra{\ss}e~3, 38106 Braunschweig, Germany}

\author{Dirk-M. Schlingemann}
\email{d.schlingemann@tu-bs.de}
\affiliation{ISI Foundation, Quantum Information Theory Unit,\\
Viale S. Severo 65, 10133 Torino, Italy}
\affiliation{Institut f\"ur Mathematische Physik, Technische
Universit\"at Braunschweig, Mendelssohnstra{\ss}e~3, 38106 Braunschweig, Germany}

\author{Zoltan Zimboras}
\email{d.schlingemann@tu-bs.de}
\affiliation{ISI Foundation, Quantum Information Theory Unit,\\
Viale S. Severo 65, 10133 Torino, Italy}

\pacs{03.67.-a, 02.30.Tb}

\date{\today}
\begin{abstract}
The mean field fluctuations of large atomic ensembles can behave like bosonic
modes, i.e. they induce a state on an appropriate system of bosonic modes.
The most prominent example is that, if the atomic ensemble is in a
homogenous product state, then the mean-field fluctuations are inducing a
Gaussian state on a system of bosonic modes. In the present paper we show that
for atomic ensemble states with exponentially decaying correlations (e.g. with
respect to the distance of atoms) the mean-field fluctuations are inducing
(possibly non-Gaussian) states on the on a system of bosonic modes. This
result is true for a general lattice of atomic systems that is equipped with a
reasonable distance function. 
\end{abstract}

%%%%%%%%%%%%%%%%%%%%%%%%%%%%%%%%%%%%%%%%%%%%%%%%%%%%%%%%%%%%%%%%%%%%%%%%%%%%%%%%%%%%%%%%%%%%%%%%%%%%
\maketitle

\section{Introduction}
\label{sec:introduction}
In many body quantum systems, the mean-field theory is focusing on the
average behavior of single constituents (see e.g. \cite{RaggWer89,RaggWer91,Werner92}). Large atomic ensembles are
systems of this kind. Typical global states of these systems have the property
that the state restricted to an single atom is independent of the individual
atom. If we average the state restricted to a single atom over all atoms
of the ensemble, we obtain the same state as restricted to each individual
atom. For asymptotically large systems, this mean-field average corresponds to
an effective classical systems. Since one is only concerned with expectation
values of observable of single atoms, correlations between different atoms are
irrelevant for the mean-field limit. 

A kind of ``first-order correction'' to the classical mean-field limit are
``mean-field fluctuations''. Observables which for testing mean-field
fluctuations are build as follows: One looks at the deviation of a single atom
observable from its mean-field expectation value. Whereas the mean-field
expectation value is the same for each single atom, the expectation value of
the deviation may depend on the individual atom. A ``mean-field fluctuation
observable'' (fluctuation operator) is an appropriate average of the
individual mean-field deviations over all atoms. 

It is well known, that if a large atomic ensemble is prepared in a homogenous
product state, i.e. each single atom is individually prepared in the same
state, the mean-field fluctuations effectively behave like a systems of
non-interacting bosonic modes. In other words, a homogenous product
state of a large atomic ensemble, induce (via mean-field fluctuations) a ``Gaussian state'' on a system of bosonic
modes. This statement has to be interpreted in the limit of infinitely large
systems. This is also related to the well known ``Holstein-Primakoff transformation'' \cite{PhysRev.58.1098}
which relates large spin systems to bosonic systems.

The ``bosonic nature'' of mean-field fluctuations for a large atomic
ensemble can also be interpreted as ``simulating'' bosonic systems by large
atomic ensembles. 

It can be observed in experiments that mean-field fluctuations can have a ``bosonic behavior'' 
by building interfaces between atomic ensembles and light
\cite{RevModPhys.82.1041}. 
Here, a laser is appropriately interacting with a gas
of atoms confined to a glass box at room temperature. The state of the laser
field can be stored into the atomic ensemble by using the effective degrees of
freedom of the mean-field fluctuations. Conversely, one can also perform an
inverse process, by transferring the state of the mean-field fluctuations of
an atomic ensemble to bosonic modes of a light field. 

One can also imagine to use a similar 
technique in order to store the state of a laser field into an ensemble of
atoms that are trapped by the periodic
potential of an optical lattice. The interesting point is here, that the implementation of
quantum cellular automata for optical lattice systems is a natural task. In an
optical lattice, each atoms occupies a lattice site. A quantum cellular
automaton is a global (mostly reversible) quantum operation whose local action
on a single site subsystem only affects the neighboring sites \cite{SchuWer04}.  

A very interesting issue combines the simulation of bosonic modes by atomic
ensembles, on one hand, with the implementation of quantum cellular automata which acts on
atomic ensembles, on the other hand: First, store the state of a laser field 
into an ensemble of atoms (arrange within a optical lattice), then implement
an quantum cellular automaton acting on the atomic ensemble, and finally
``release the light'' from the atomic ensemble. By the overall process, we
obtain an incoming light field and a outgoing scattered light field. The
question is now: 

\begin{quote}
What kind of effective operation is describing this ``scattering process''?
\end{quote}

We expect to obtain operations that go beyond the ``Gaussian world'' 
which can be used as ``non-Gaussian addons''. Since Gaussian operations are
limited in their ability to perform quantum information tasks, this may open a
door to perform new tasks.  

Whether the mean-field fluctuations of a large atomic ensemble behave like
bosonic modes depends on the state of the atomic ensemble. If we want to find
an answer to the question given above, we will answer the following 
question first:

\begin{quote}
For which states of large atomic ensembles do the mean-field fluctuations
behave like bosonic modes and is there a set of states with bosonic mean
field fluctuations which is invariant under application of quantum cellular automata? 
\end{quote}  

We shall see that Theorem~\ref{thm-main} provides an answer to this question. 
For states of large atomic ensembles whose correlations are exponentially
decaying with the distance of the single atoms (exponential clustering), the mean-field fluctuations
behave like bosonic modes. In particular, states with exponential clustering
are invariant under quantum cellular automata since the action of a quantum
cellular automaton on a single site system only affects a finite set of neighbors. 
  
As a consequence, the following process is possible: A laser field is
interacting with a large atomic ensemble such that the Gaussian state of the laser field 
is encoded into a homogenous product state of the atoms. A quantum cellular
automaton acting on the atomic ensemble is implemented. The resulting state of
the atomic ensemble possesses again bosonic mean-field fluctuations.  Finally,
we can use again the interaction between the laser field and the atomic
ensemble to transfer the state of the atomic ensemble (almost faithfully) to  
the bosonic modes of the laser.  

The total process induces an operation on bosonic modes which maps an initial
Gaussian state to some bosonic state, which can be non-Gaussian. With help of
Theorem~\ref{prop-expansion}, the correlation functions of the resulting state
can be written as a perturbation of a Gaussian state. This may be helpful in
oder to decide which states of atomic ensembles correspond to Gaussian states.

\subsection*{Outline of the paper}

In Section~\ref{sec:univ-descr-cont} we provide the appropriate mathematical
tools for describing mean-field fluctuation. Fixing a given atomic ensemble,
it depends on the state which kind of system of bosonic modes (if there is
one) is corresponding to the mean-field fluctuations. This requires to compare
different bosonic systems even if they differ by its  canonical commutator
relations. The tensor algebra provides a universal description that covers all
different bosonic systems at once. Which bosonic system is realized is part of
the state of this ``universal continuous variable system''.

How to describe and analyze mean-field fluctuation by using the framework of
universal continuous variable systems (tensor algebras) is discussed in 
Section~\ref{sec:mean-field-fluct}. Here, we also present the main results
(Theorem~\ref{prop-expansion}, Theorem~\ref{thm-main}). 
Technical supplements in order to give self-contained proofs are postponed to
the sections in the appendix.

\subsection*{Acknowledgment}
This work was supported by the EU FP7 FET-Open project COQUIT (contract
number 233747).

\section{Universal description of continuous variable systems}
\label{sec:univ-descr-cont}
Our goal is to relate systems of atomic ensembles with
bosonic systems. In this section, we explain this relation in mathematical
detail and generality, by using the algebraic approach to quantum mechanics.
Here systems are given in terms of their observable algebras, which, in our
case, are C*-algebras or more general *-algebras when unbounded operators are
included.  

\subsection{The tensor algebra a universal playground} 
We now introduce a formalism for describing general continuous variable
systems in a uniform manner, which is not so frequently used, but which has the advantage that the
Holstein-Primakov transformation can be implemented easily and naturally.

Let $V$ be a complex vector space with a complex conjugation $J$. The tensor
algebra over $(V,J)$ is the unital associative *-algebra that is given by the complex
vector space 
\begin{equation}
\9T(V,J)=\bigoplus_{n\in\7N} V^{\otimes n} \; .
\end{equation}
The product, which is just given by the tensor product, is determined by  
\begin{equation}
(v_1\otimes v_2\otimes \cdots\otimes v_n)(w_1\otimes w_2\otimes\cdots\otimes
w_m)= v_1\otimes v_2\otimes \cdots\otimes v_n\otimes w_1\otimes w_2\otimes\cdots\otimes
w_m 
\end{equation}
and the adjoint is determined by 
\begin{equation}
(v_1\otimes v_2\otimes \cdots\otimes v_n)^*=Jv_{n}\otimes
Jv_{n-1}\otimes\cdots\otimes Jv_1 \; ,
\end{equation}
where $v_1,\cdots, v_n,w_1,\cdots, w_m\in V$. Note that $V^{\otimes 0}\cong
\7C$ corresponds to the multiples of the unit operator $\11$. Obviously, there
is a linear embedding $\Phi$ of $V$ into the tensor algebra $\9T(V,J)$ such
that $\Phi(v)^*=\Phi(Jv)$. In the following, we call the operators $\Phi(v)$
``generalized field operators''. 

The tensor algebra $\9T(V,J)$ represents the
observable algebra for a wider class of continuous variable systems. The detailed type of the
system, e.g. fermionic or bosonic, is encoded in the states under
consideration.  
A state $\omega$ is described by normalized positive linear functional on the
tensor algebra, i.e. $\omega(A^*A)\geq 0$ and
$\omega(\11)=1$. Each state is determined by the $n$-point correlation
functions 
\begin{equation}
\omega_n(v_1,\cdots ,v_n)=\omega(v_1\otimes\cdots\otimes v_n)
\end{equation}
where $\omega_n$ is a $n$-multi-linear functional on $V$.

How is all this related to the
ordinary Hilbert space formalism of quantum mechanics? Well, with help of the
so called GNS representation we obtain a Hilbert space $\2H_\omega$, a vector $\Omega_\omega$
and a *-representation $\pi_\omega$ by linear (but
not necessarily bounded) operators on $\2H_\omega$ with
$\omega(A)=\SPROD{\Omega_\omega}{\pi_\omega(A)\Omega_\omega}$. This is just a
consequence of the positivity of the functional $\omega$. 

\subsection{Realizing bosonic systems}
As a first example, let us have a look at the bosonic systems. For this purpose we
construct a``quasi-free bosonic state'' from a bilinear form, called the covariance 
$W$, on $V$. In order to obtain a positive functional, the positivity
condition $W(Jv,v)\geq 0$ has to be fulfilled.  The corresponding quasi-free
state is determined according to the following conditions: For $n>0$ we put 
\begin{equation}
\label{eq:2}
\omega_W(v_1\otimes\cdots\otimes v_n):=\sum_{P\in \Pi_2(n)} \prod_{(i,j)\in P} W(v_i,v_j)
\end{equation}  
and $\omega_w(\11)=1$, where $\Pi_2(n)$ is the set of ordered partitions of
$\{1,\cdots ,n\}$ into two-elementary subsets. Note that the sum is empty for
odd $n$. 

We associate to $W$ the hermitian form $\gamma$ which is given by
$\gamma(v_1,v_2)=W(Jv_1,v_2)-W(v_2,Jv_1)$. Using the Araki's self-dual
formalism, the self-dual CCR algebra $\8{CCR}(V,J,\gamma)$ is the *-algebra
that is constructed as follows: Let $\9J_\gamma$ be the two sided ideal in
$\9T(V,J)$ that is generated by the operators
$\Phi(v)^*\Phi(v')-\Phi(v')\Phi(v)^*-\gamma(v,v')\11$. Then the corresponding
self-dual CCR algebra is given by the quotient *-algebra   
\begin{equation}
\8{CCR}(V,J,\gamma):=\9T(V,J)/\9J_\gamma \; .
\end{equation}
As we will briefly sketch below, the state $\omega_W$ annihilates the ideal
$\9J_\gamma$, which implies that  $\omega_W$ induces a unique state on
$\8{CCR}(V,J,\gamma)$. Two quasi-free states $\omega_W$ and $\omega_{W'}$ on
the tensor algebra belong to the same Bosonic system if $W(Jv,v')-W(v',Jv)=
W'(Jv,v')-W'(v',Jv)=\gamma(v,v')$. In this case, both states annihilates the
ideal $\9J_\gamma$ and can be lifted to the same CCR algebra. 

\begin{rem}\em
The quasi-free states $\omega_W$ have the special property to be ``even'',
i.e. the expectation value of a single generalized field operator is vanishing
$\omega_W(\Phi(v))=0$. To obtain all quasi-free states, we take advantage of
the following fact: Let $V_{\7R}^*$ be the real vector space of real continuous linear functionals on
$V$. Note that a functional $u\in V^*$ is real if it fulfills the condition
$u(Jv)=\overline{u(v)}$ for all $v\in V$. If we regard $V_{\7R}^*$ with its addition as an Abelian group, then $V^*_{\7R}$  
is acting by *-automorphisms on the tensor algebra $\9T(V,J)$. For each $u\in
V_{\7R}^*$, we define the *-automorphism $\alpha_u\in\8{Aut}(\9T(V,J))$ according to 
\begin{equation}
\alpha_u\Phi(v):=\Phi(v)+u(v)\11 \; .
\end{equation}
By construction, the group law is fulfilled, i.e. $\alpha_{u_1}
\alpha_{u_2}=\alpha_{u_1+u_2}$ is valid for all $u_1,u_2\in V^*_{\7R}$. To
obtain a quasi-free state with a non-vanishing one-point function, we just
``shift'' an even quasi-free state $\omega_W$ by an appropriate automorphism
$\alpha_u$ yielding the quasi-free state $\omega_{W,u}=\omega_W\circ \alpha_u$
which has the one-point function
$\omega_{W,u}(\Phi(v))=\omega_W(\Phi(v))+u(v)=u(v)$.  
\end{rem}

\subsection{Ideals to specify more detailed systems}
The discussion of the previous subsection shows that the tensor algebra can be
indeed used to describe various
systems by one unified object. To specify a more particular sub-class of
systems additional algebraic relations has to be respected. This corresponds
to a proper two-sided ideal $\9J\subset\9T(V,J)$. By inclusion, the set of
two-sided ideals is partially ordered. As larger the ideal, as more specific
is the systems class under consideration. For instance, if the hermitian form
$\gamma$ is non-degenerate, then the ideal $\9J_\gamma$ which describes the
corresponding CCR-relations is maximal: This can be interpreted as the most
specific description of a system, here for a set of bosonic modes.  Each state 
$\omega$ is accompanied with a natural system that is 
given by the quotient algebra $\9A_\omega:=\9T(V,J)/\9J_\omega$, where
$\9J_\omega$ is the two-sided ideal $\9J_\omega:=\{A|\forall B,C:\omega(B^*AC)=0
\}$. Note that, by construction,
$\9J_\omega$ does not contain the identity operator and is therefore a proper
ideal.    

\subsection{Comparison of states}
But what does it mean, that two states on the tensor algebra are close to each
other? To give a precise answer to this question, we need to compare states
quantitatively. For this purpose, we assume that $V$ is a Banach space.
The dual space of the tensor algebra $\9T(V,J)$ is denoted by $\9T(V,J)^*$. It
consists of all linear functionals $F:\9T(V,J)\to \7C$ for which for all $n\in
\7N$ the semi-norms
\begin{equation}
\nu_n(F):=\sup_{(v_1,\cdots,v_n)\in V_1^n}|F(v_1\otimes\cdots\otimes v_n)| <\infty
\end{equation} 
are bounded, where $V_1=\{v\in V|\|v\|=1\}$ is the unit sphere. Now, $\9T(V,J)^*$ is closed in the following topologies:
\begin{itemize}
\item
The
{\em strong topology} is the locally convex topology that is induced by the
family of semi-norms $\nu_n$, $n\in \7N$. 
\item
The {\em weak topology} is the locally convex topology that is induced by the
family of semi-norms
$\nu_{(v_1,\cdots,v_n)}(F):=|F(v_1\otimes\cdots\otimes v_n)|$ with  $(v_1,\cdots,v_n)\in\bigcup_{k}V^k$.
\end{itemize}

From an experimental perspective, the strong topology is related to the  
comparison of two states $\omega$ and
$\omega'$. Suppose we  estimate for a finite family of vectors $v_1,\cdots,
v_n\in V$ 
the correlation functions 
$\omega(v_1\otimes\cdots\otimes v_n)$ and 
$\omega'(v_1\otimes\cdots\otimes v_n)$. Then the modulus of the difference 
of the correlation functions can be used as a measure how 
``close'' $\omega$ and $\omega'$ are to each other. 

To give an example, we consider the correlation functions of two
quasi-free states $\omega_W$ and $\omega_{W'}$, where the covariances $W,W'$
have a norm difference that is given by $\|W-W'\|=\sup_{v_1,v_2\in V_1}|W(v_1,v_2)-W'(v_1,v_2)|$.    

\begin{prop}
\label{prop:quasi-free-contin}
Let  $\omega_W$ and $\omega_{W'}$ be quasi-free states with covariances $W$
and $W'$ respectively, then for each $n\in\7N$ the semi-norm difference of the
quasi-free states satisfies the bound 
\begin{equation}
\nu_n(\omega_W-\omega_{W'})\leq\|W-W'\| \ 
|\Pi_2(n)|\sum_{k=1}^{n/2} \ \|W\|^{k-1}\|W'\|^{n/2-k} \; .
\end{equation} 
\end{prop}

A direct consequence of the proposition (which we prove in the appendix) is that, if $W\to W'$ are converging
in norm, then $\omega_W\to\omega_{W'}$ converges in the strong topology. In
other words the mapping $W\to\omega_W$ is continuous in the respective
topologies. 
 
\section{Mean-field fluctuations}
\label{sec:mean-field-fluct} 
Many systems under consideration possessing a large number of independent degrees of
freedom such that they can be idealized by infinite systems in the
thermodynamic limit. Here we model this situation by an infinite (countable) 
lattice $\Lambda$ that possesses a distance function
$d:\Lambda^2\to\7R_+$. The observable algebra of the global system is the
so called {\em quasi-local algebra} that is constructed by the infinite tensor
product 
\begin{equation}
\6A(\Lambda)=\bigotimes_{x\in\Lambda}\6A(\Lambda)
\end{equation}
of single cell C*-algebras $\6A\cong\6A(x)$. For a given lattice point $x\in X$, the natural embedding 
of the single site algebra which identifies $\6A$ with $\6A(x)\subset\6A(\Lambda)$
is denoted by $\iota_x$. We are going to use this mapping later on quite often.

However, in view of non-equilibrium thermodynamics, the nature of global states of an
infinite systems can be very different from equilibrium states and
the calculation of expectation values for such states may be a hard
computational task. The analysis of global states, one looks at the
asymptotic behavior of certain properties within the mesoscopic range. For
this purpose, one restricts the global state to the local observable algebras
that correspond to finite subsets sets $X\subset\Lambda$ which 
is given by the finite tensor product
\begin{equation}
\6A(X)=\bigotimes_{x\in X}\6A(x) \, .
\end{equation}
Note that for an inclusion $X\subset Y\subset\Lambda$, it follows immediately
that $\6A(X)\subset \6A(Y)$.  

Taking a global state $\omega_\Lambda$, we obtain for each finite subset
$X\subset \Lambda$ a restricted state $\omega_X:=\omega_\Lambda|_{\6A(X)}$
which lives on finitely many degrees of freedom. This yields a net of sates
$(\omega_X)_{X\subset\Lambda}$ that is indexed by the partially ordered set of
finite subsets of the lattice $\Lambda$. Roughly speaking, the basic idea behind the
Holstein-Primakov transformation is to analyze the behavior of
each of the states $\omega_X$ concerning their ``bosonic nature'', i.e. to
what extend they ``simulate'' continuous variable systems. We shall see, that
each restriction $\omega_X$ induces a state $\hat\omega_X$ on the tensor
algebra $\9T(\6A,*)$, where $\6A$ is the observable algebra of a single cell
system. To be of ``bosonic nature'', the induced state $\hat\omega_X$ has to
fulfill ``almost'' the canonical commutation relations. This means that there
is an antisymmetric hermitian form $\gamma$ such that the induced state
$\hat\omega_X$ is ``almost'' annihilating the ideal $\9J_\gamma$: A typical
behavior is $\hat\omega_X(A)=O(|X|^{-1/2})$ for each operator $A$ that belongs
to the ideal $\9J_\gamma$.

\subsection{Inducing states and $\sqrt{\rm n}$-fluctuations}
Let $\omega_\Lambda$ be a state of the global system. Then we obtain the net
of restricted states $(\omega_X)_{X\subset\Lambda}$ that are indexed by the
partially ordered set of finite subsets $X\subset \Lambda$. The induction of
states works by using ``fluctuation operators'' 
associated with the restricted state $\omega_X$ and an operator $a\in\6A$:   
\begin{equation}
\Phi_{\omega_X}(a):=\frac{1}{|X|^{1/2}}\sum_{x\in X}
[\iota_{x}a-\omega_X(\iota_{x}a)\11] \; .
\end{equation}  
This yields a representation $\Phi(a)\mapsto\Phi_{\omega_X}(a)$ of the tensor
algebra and the induced state 
$\hat\omega_{X}$ is determined by its $n$-point functions according to 
\begin{equation}
\hat\omega_{X}(a_1\otimes\cdots\otimes a_n):=\omega_X(\Phi_{\omega_X}(a_1)\cdots\Phi_{\omega_X}(a_n))
\; .
\end{equation}  

The main goal is now to study the asymptotic limit of large systems. For this
purpose, let $\Lambda$ be a countable lattice.  Let $(\omega_X)_{X\subset
  \Lambda}$ be a net of states that is indexed by finite subsets of
$\Lambda$, where $\omega_X$ is a state on $\6A(X)$. The asymptotic properties
in the limit $X\to\Lambda$ can be investigated by looking at the net of
induced states $(\hat\omega_{X})_{X\subset\Lambda}$ according to the classification:
    
\begin{itemize}
\item
The state $\omega_\Lambda$  has
$\sqrt{\rm n}$-fluctuations if  the induced net
$(\hat\omega_{X})_{X\subset\Lambda}$ converges 
$w-\lim_{X\to\Lambda}\hat\omega_{X}=\hat\omega_{\Lambda}$ in the weak
topology on $\9T(\6A,*)$.
\item 
The state $\omega_\Lambda$  has
strongly $\sqrt{\rm n}$-fluctuations if  the induced net
$(\hat\omega_{X})_{X\subset\Lambda}$ converges 
$s-\lim_{X\to\Lambda}\hat\omega_{X}=\hat\omega_{\Lambda}$ in the strong
topology on $\9T(\6A,*)$.
\item
The state $\omega_\Lambda$  has
weakly $\sqrt{\rm n}$-fluctuations if for the induced net
$(\hat\omega_{X})_{X\subset\Lambda}$ the each semi-norm $\nu_n$, $n\in \7N$, is uniformly bounded:  
$\sup_{X\subset\Lambda}\nu_n(\hat\omega_{X})<\infty$.
\end{itemize}

Obviously, strongly $\sqrt{\rm n}$-fluctuations implies $\sqrt{\rm
  n}$-fluctuations implies weakly $\sqrt{\rm n}$-fluctuations.
We are now considering states that are ``single site homogenous''. 
These states are defined by the property that their restrictions 
to a single site is independent of the
lattice point.

\subsection{Induced states for asymptotically large systems}
Asymptotically large atomic ensembles can be described by an infinite lattice
system which is in some state $\omega_\Lambda$. Suppose we assume that the
corresponding induced net of states  $(\hat\omega_{X})_{X\subset\Lambda}$  has
weakly $\sqrt{\rm n}$-fluctuations. What conclusions can we draw from this
property? What do we know about the asymptotic behavior of the correlation
functions of the induced states $\hat\omega_{X}$?

Since each semi-norm $\nu_n(\hat\omega_X)$ is uniformly bounded in the size
of the subset $X$, we know that there are weak limit points. In order to
analyze the properties of these limit points more systematically, we will give 
here an ``operational'' description of what limit points are. 

Within a concrete experimental realization, the atomic ensemble under consideration
will be always finite. If the setup is scalable, then, at least in principle,
the same experiment can be performed for various sizes of the system, i.e. the
subset $X$ can be regarded as a ``classical configuration''. Here one can also
think of a situation, where atoms occupy only finitely many sites of a lattice
randomly. Thus we are dealing with a preparation device that prepares for each 
finite atomic ensemble $X\subset \Lambda$ a state $\omega_X$ with a certain
probability $\mu(X)$. The probability distribution $\mu:X\mapsto \mu(X)$ is nothing
else but a classical state on the system of finite subsets
$X\subset \Lambda$. The corresponding observable algebra consists of all bounded
complex valued functions $f:\Lambda\supset X\mapsto f(X)$. The expectation
value of an observable $f$ for the state $\mu$  is then  
given by $\mu(f)=\sum_{X\subset\Lambda}\mu(X) f(X)$. 

A general classical state
on the system of finite subsets $X\subset\Lambda$ is a complex valued linear
functional on the algebra of  bounded complex valued functions 
$f:X\mapsto f(X)$ such that the following holds:
\begin{itemize}
\item
Positivity: $\eta(f)\geq 0$ for each $f\geq 0$. 
\item
Normalization: $\eta(\11)=1$.
\end{itemize}

A preparation device that produces asymptotically large atomic ensembles has
the property that, in the limit $X\to\Lambda$, the probability that only a
finite number of lattice sites are occupied is vanishing. This corresponds to
classical states $\eta$ with the following property:   
\begin{itemize}
\item
$\eta(f)=0$ if $\lim_{X\to\Lambda}f(X)=0$.
\end{itemize}
A state $\eta$ with this property is called a ``limit point''. To justify this notion,
suppose that limit $\lim_X f(X)=c$ exists. In this case
$\lim_X(f(X)-c)=0$, and $\eta(f-c\11)=\eta(f)-c=0$
follows, which means that there expectation value $\eta(f)=c$ coincides for
all limit points $\eta$. 

What can we say about the limit points of the induced net
$(\hat\omega_X)_{X\subset\Lambda}$ for a state $\omega_\Lambda$ 
that have weakly $\sqrt{\rm n}$-fluctuations? To each operator
$A\in\9T(\6A,*)$ of the tensor algebra, we assign a bounded function
which is given by $X¸\mapsto \hat\omega_X(A)$ \footnote{This function is
indeed bounded which can be verified as follows: The operator $A$ can be
written as a finite direct sum $\bigoplus_{k=0}^nA_k$ with $A_k\in\6A^{\otimes
  k}$.  Since $\omega_\Lambda$ has weakly $\sqrt{\rm n}$-fluctuations we obtain
that $|\hat\omega_X(A)|\leq \sum_{k=0}^n|\hat\omega_X(A_n)|\leq \sum_{k=0}^n
C_n \|A_n\|<\infty$ with $C_n=\sup_{X\subset\Lambda}\nu_n(\hat\omega_X)$.}.  
Now, each limit point $\eta$ induces a state on the tensor
algebra by $\omega_\eta(A)=\eta(X\mapsto \hat\omega_X(A))$. Here we use the
suggestive notation $\eta(X\mapsto f(X)):=\eta(f)$ to represent an expectation value.    

The states $\hat\omega_\eta$ describe the mean-field fluctuations
of asymptotically large
atomic ensembles. The next propositions states that these mean-field
fluctuations behave like bosonic modes. Consider a state $\omega_\Lambda$ that
is single site homogenous with single site restriction
$\omega=\omega_\Lambda\circ \iota_x$. 
Then there is a natural antisymmetric 
hermitian form $\gamma(a,b):=\omega([a^*,b])$ on the observable algebra $\6A$ of the
single site system. The ideal $\9J_{\gamma}$, which represents the canonical
commutation relations, is generated by the operators  
$I_\gamma(a,b):=[\Phi(a),\Phi(b)]-\gamma(a^*,b)\I$. 

\begin{prop}
\label{prop-bosonic}
Let $\omega_\Lambda$ be a single site homogenous state having weakly $\sqrt{\rm n}$-fluctuations.
Then for each limit point $\eta$, the state $\hat\omega_\eta$ annihilates the
ideal $\9J_{\gamma}$ and can uniquely be lifted to a state on the
corresponding CCR algebra.
\end{prop}
\begin{proof}
By Lemma~\ref{lem-00} of the appendix, we conclude that $\lim_X\hat\omega_X(A)=0$ for each
operator in the ideal $\9J_{\gamma}$. This implies
$\hat\omega_\eta(A)=\eta(X\mapsto\hat\omega_X(A))=0$ which implies that
$\hat\omega_\eta$ annihilates $\9J_\gamma$.
\end{proof}
  
\begin{rem}\em
Proposition~\ref{prop-bosonic} can be interpreted, at least to a certain extend, 
by saying that a state $\omega_\Lambda$ of a large atomic ensembles
with weakly $\sqrt{n}$-fluctuations possess bosonic mean-field
fluctuations. This is justified by the fact that each limit point
$\hat\omega_\eta$ is a state on the CCR algebra $\8{CCR}(\6A,*,\gamma)$ which
describes a bosonic system. On the other hand, the CCR algebra is an algebra of
unbounded operators and it might happen that the GNS representation associated
to a limit state $\hat\omega_\eta$ has ``exotic'' properties. Recall, that the GNS
representation is given by a Hilbert space $\2H$ an algebra homomorphism $\pi$
that assigns to each operator $A$ in the CCR algebra a linear (unbounded)
operator on $\2H$ as well as a normalized vector $\Omega\in\2H$ such that
$\omega_\eta(A)=\langle \Omega, \pi(A)\Omega\rangle$. The question that arises
here is whether it is possible to build the exponential
$\exp(\8i\pi(\Phi(a)))$ of a field operator $\pi(\Phi(a))$ in the
representation $\pi$, where $a=a^*$ is selfadjoint. If we can do this, then we
obtain a representation of the Weyl algebra by bounded operators. If it is not
possible to build the exponential we are dealing with an 
``exotic'' case (see e.g. \cite{RESI1}). In order to exclude this kind of
pathologies, we need to consider more specific examples of states.
\end{rem}

\subsection{States with exponential clustering}
A state $\omega_\Lambda$ on $\6A(\Lambda)$ has exponential clustering (with
respect to $d$) if for local operators $A\in\6A(X)$ and $B\in\6A(Y)$ the identity  
\begin{equation}
\omega_\Lambda(AB)=\omega_\Lambda(A)\omega_\Lambda(B)+G_{(X,Y)}(A,B)\8e^{-d(X,Y)}
\end{equation}
is valid for a bounded bilinear function $G_{(X,Y)}:\6A(X)\times\6A(Y)\to\7C$ such that
$|G_{(X,Y)}(A,B)|\leq G_0 \|A\|\|B\|$ for all $A,B$, for all finite regions $X,Y$. Here $G_0$ is
a constant that is independent of the localization regions. The bilinear forms
$G_{X,Y}$ express locally the deviations from the state to a product state,
being scaled with the exponential of the distance. Therefore $G_{(X,Y)}$
indicates the presence of correlations that are exponentially decreasing with
the distance. To give a name, we call the family of bilinear maps
$G=(G_{(X,Y)})_{X,Y\subset \Lambda}$ the {\em correlators}. 
Note that, equivalently, exponential clustering is given by the condition
\begin{equation}
|\omega_\Lambda(AB)-\omega_\Lambda(A)\omega_\Lambda(B)|\leq G_0 \ \8e^{-d(X,Y)}
\end{equation}
for all $A\in\6A(X),B\in\6A(Y)$. Here $d(X,Y)=\min_{x\in X,y\in Y}d(x,y)$ is
the distance between the finite subsets $X,Y\subset\Lambda$. We always require
here, that the distance $d$ is regular, i.e. the maximal number $N(r)$ of lattice sites
within a ball of radius $r$ is bounded by a polynomial.

The exponential clustering property can be used to derive a useful cluster expansion
in terms of expectation values of the single site restriction $\omega$ and the correlators
$G$. In order to write down this expansion, we introduce the following
objects:
\begin{itemize}
\item
For each finite subset $Y\subset \Lambda$ we introduce the
``the spread'' $\Delta(Y):=\max_{y\in Y}d(y,Y\setminus y)$ which measures the
maximal distance of a point to its relative complement in $Y$. 

\item
An $k$-elementary subset
$\{y_1,\cdots,y_k\}\subset  X$ is called ``spread optimally enumerated'' if the enumeration
fulfills the condition
$d(y_l,\{y_{l+1},\cdots,y_k\})=\Delta(y_l,\cdots,y_k)$ for all
$l=1,\cdots,k-1$. Note that each subset can be spread optimally
enumerated. 

\item
Given a tuple $x\in X^n$, we choose a spread optimal enumeration
of the range $\8{Ran}(x)=\{y_1,\cdots,y_{|\8{Ran}(x)|}\}$ and we consider the
correlators $G^x_k:=G_{(y_k,\{y_{k+1},\cdots,y_{|\8{Ran}(x)|}\})}$ which test
the correlations for splitting the site $y_k$ from the remaining points
$\{y_{k+1},\cdots,y_{|\8{Ran}(x)|}\}$, where $k=1,\cdots,|\8{Ran}(x)|-1$. 

\item 
For a family of operators $a_1,\cdots,a_n\in\6A$ and a tuple $x\in X^n$ whose
range $\{y_{1},\cdots,y_{|\8{Ran}(x)|}\}$  is spread optimally enumerated, we
introduce the single site ``cluster operators'' $a^x_{k}\in\6A$ which are  given by the ordered
product $a_k^x:=\prod_{j\in x^{-1}(y_k)}a_j$, where the ordering is according
to the value of the index in $x^{-1}(y_k)=\{j=1,\cdots,n|x_j=y_k\}$.
\end{itemize}

The following theorem, whose proof is given in the appendix, states that
correlation functions of the induced states $\hat\omega_X$ admit a cluster
expansion in terms of the single site restriction $\omega$, the correlators
$G$ and cluster operators $a_k^x$: 
 
\begin{thm}[Cluster expansion]
\label{prop-expansion}
Let $\omega_\Lambda$ be a single site homogenous state with single site
restriction $\omega$ and exponential
clustering with respect to $d$. For each $a_1,\cdots,a_n\in\8{ker}(\omega)$
and for each finite subset $X\subset \Lambda$  the $n$-point correlation 
function of the induced state $\hat\omega_X$ can be written as
\begin{equation}
\label{equ-expansion}
\begin{split}
&\hat\omega_X(a_1\otimes\cdots\otimes a_n)=\hat\omega^{\otimes
  X}(a_1\otimes\cdots\otimes a_n)+F_X(a_1\otimes\cdots\otimes a_n) \; .
\end{split}
\end{equation}
where the correlation function of the induces homogenous product state
$\hat\omega^{\otimes X}$ and the functional $F_X$ are given by 
\begin{equation}
\label{eq:1}
\begin{split}
\hat\omega^{\otimes X}(a_1\otimes\cdots\otimes a_n)
&=|X|^{-\frac{n}{2}}\sum_{x\in X^n}\omega(a_{1}^x)\cdots\omega(a_{|\8{Ran}(x)|}^x)
\\
F_X(a_1\otimes\cdots\otimes a_n)
&=|X|^{-\frac{n}{2}}\sum_{x\in X^n}\sum_{k=1}^{|\8{Ran}(x)|-1} \omega(a_{1}^x)\cdots\omega(a_{k-1}^x)
\\
&\times \ \  G^x_k(a^x_k,a^x_{k+1}\cdots
a^x_{|\8{Ran}(x)|})\8e^{-\Delta(y_k,y_{k+1},\cdots,y_{|\8{Ran}(x)|})} \; ,
\end{split}
\end{equation}
where for each $x\in X^n$ the range $\8{Ran}(x)$ is spread optimally enumerated.
\end{thm}

As the cluster expansion is stated above, it holds for all correlation
functions for which the operators $a_1,\cdots,a_n$ are chosen in the kernel
$\8{ker}(\omega)$ of the single site restriction. If this is not the case, we
can express the correlation function in terms of $a_i=a_i'+\omega(a_i)\11$ where
$a_i'\in\8{ker}(\omega)$. The tensor product $a_1\otimes\cdots \otimes a_n$
can be expanded in terms of the operators $a_i'\in \8{ker}(\omega)$  according to 
\begin{equation}
\label{eq:3}
a_1\otimes\cdots \otimes a_n=\sum_{J\subset \{1,\cdots, n\}} \bigotimes_{i\in J}
a'_{i} \prod_{j\in \{1,\cdots, n\}\setminus J} \omega(a_j) 
\end{equation}
where the sum runs over all ordered subsets.  To get the general 
cluster expansion for the full tensor algebra, one only has
to apply  Theorem~\ref{prop-expansion} to (\ref{eq:3}) for each summand. 

It is known, that the induced net $(\hat\omega^{\otimes X})_{X\subset\Lambda}$
converges weakly to a quasi-free state. 
We show here a slightly stronger result:

\begin{prop}
\label{prop-prod}
A homogenous product state $\omega^{\otimes\Lambda}$ has strongly
$\sqrt{\rm n}$-fluctuations. In particular, the induced net 
$(\hat\omega^{\otimes X})_{X\subset\Lambda}$ converges strongly to the quasi-free state 
$\hat\omega_{\8{qf}}$ whose covariance is given by the truncated two-point
function $W(a,b)=\omega(ab)-\omega(a)\omega(b)$. 
\end{prop}

Homogenous product states are the simplest among states that have exponential
clustering. For the general case, the following is true:

\begin{thm}
\label{thm-main}
Each single site homogenous state with exponential clustering has weakly 
$\sqrt{\rm n}$-fluctuations. 
\end{thm}

The proof of the theorem is quite technical and therefore postponed to the
appendix. However, it takes advantage of the cluster expansion of $F_X$ into single
site expectation values and correlators. The basic idea to get a uniform bound
for the semi norms $\nu_n(F_X)$ is to count the number of terms that are
contributing to the cluster expansion. In total, we sum over all tuples in $X^n$ which
gives $|X|^n$ terms. Since we normalize by multiplying $|X|^{-n/2}$, a naive
counting would give the non-uniform bound $\nu_n(F_X)\leq \8O(|X|^{-n/2})$. By a
more careful analysis, it turns out that effectively only $|X|^{n/2}$ terms
are contributing. By choosing $a_1,\cdots,a_n\in\8{ker}(\omega)$, 
the single site expectation value of a cluster operator 
$\omega(a_k^x)$ is vanishing if $x^{-1}(y_k)=\{j\}$ contains only a single
element. Note that  in this case we just have
$\omega(a^x_k)=\omega(a_j)=0$. This reduces directly the
number of terms which in the cluster expansion (\ref{eq:1}). A large
number contributions are also coming from tuples $x$ with range
$\{y_1,\cdots,y_{|\8{Ran}(x)|}\}$ for which the
spreads $\Delta(y_k,\cdots, y_{|\8{Ran}(x)|})$ are large. These
contributions are also of order $|X|^{n/2}$, since they are suppressed the
exponential damping $\exp(-\Delta(y_k,\cdots, y_{|\8{Ran}(x)|}))$.   

\begin{rem}\em 
We can derive from the cluster expansion that
for asymptotically large atomic ensembles there is a quasi-free part from the
product state contribution and a perturbation which comes from the correlators.   
Namely, for each weak limit point $\eta$ the state
$\hat\omega_\eta$ can be written as
\begin{equation}
\hat\omega_\eta=\hat\omega_{\8{qf}}+F_\eta
\end{equation}
with $F_\eta(A)=\eta(X\mapsto F_X(A))$. The functional $F_\eta$ is a perturbation of the quasi-free
limit state $\hat\omega_{\8{qf}}$ which may depend on the limit functional
$\eta$. Note that Theorem~\ref{thm-main} guarantees the existence of weak
limit points $F_\eta$, 
since $\sup_{X\subset\Lambda}\nu_n(F_{X})<\infty$.
\end{rem}

\section{Conclusion}

We have shown that states of large atomic ensembles whose correlations are exponentially
decaying with the distance between atoms (exponential clustering) possess
bosonic mean-field fluctuations. In addition to that, these states 
are invariant under applications of quantum cellular. 
  
This enables the implementation of the following type of process: The bosonic
modes of a light field are coupled to a large atomic ensemble such that the
Gaussian 
state of the laser field 
is transferred almost perfectly (where the precision is here of order
$\8O(\sqrt{\mbox{number of single atom systems}})$) to a homogenous product state of the atoms. A quantum cellular
automaton acting on the atomic ensemble is implemented. The resulting state of
the atomic ensemble possesses again bosonic mean-field fluctuations and 
the resulting state of the atomic ensemble can be transferred back almost perfectly to  
the bosonic modes of the light field.  

The total process induces an operation on bosonic modes which maps an initial
Gaussian state to some bosonic state, which can be non-Gaussian. With help of
Theorem~\ref{prop-expansion}, the correlation functions of the resulting state
can be written as a perturbation of a Gaussian state. This may be helpful in
oder to decide which states of atomic ensembles correspond to Gaussian states.
 
It is still an open problem to decide in general from the state of the atomic ensemble whether the
resulting induced state is Gaussian or not. Concerning states with
exponential clustering, the cluster
expansion (Theorem~\ref{prop-expansion}) appears to be a reasonable technique in order to address this problem.      
Here the correlation functions of the fluctuation operators can be expanded
into the correlation functions of the homogenous product state
$\omega^{\otimes X}$ (here $\omega$ is the restriction of the global
state to a single atom) and some correction $F_X$. For large atomic ensembles,
the correlation functions of the homogenous product state $\omega^{\otimes X}$
correspond to a Gaussian state, whereas $F_X$ can be regarded as a
``perturbation''. 

Furthermore, it would be desirable to construct new examples of atomic
ensemble states (in particular beyond homogenous product states)  whose
induced net has strongly $\sqrt{\rm n}$-fluctuations or $\sqrt{\rm
  n}$-fluctuations.  

In order to archive more concrete results in this direction, one has to
consider here more concrete examples. One suggestion is to consider ensembles
of two-level atoms arranged in a one-dimensional lattice. A natural class of
states for which mean-field fluctuations can be investigated are stabilizer
states which are invariant under the action of so called Clifford quantum
cellular automata (see \cite{SchlVogtWer08,Schl09,Gutschow:1184639} and references given therein).

%\bibliographystyle{plain}
%\bibliography{qinf.bib}

\clearpage

\begin{appendix}

\section{On quasi-free states (proof of Proposition~\ref{prop:quasi-free-contin})}
\begin{proof}[ of Proposition~\ref{prop:quasi-free-contin}]
Let $W,W'$ be two bounded covariances (positive bounded bilinear form) on $V$. Recall that a bilinear form $F$
on $V$ is bounded if 
$\|F\|:=\sup_{v_1,v_2\in V_1}|F(v_1,v_2)|<\infty$. Moreover, recall that a valid covariance $W$
has to be positive in the sense that $W(Jv,v)\geq 0$ for all
$v\in V$. We fix vectors $v_1,\cdots,v_n$ and use  
the expression (\ref{eq:2}) for the correlation functions of quasi-free
states to calculate the difference of the $n$-point functions for the
quasi-free states $\omega_W$ and $\omega_{W'}$. For each ordered partition $P\in
\Pi_2(n)$ we choose an enumeration $P=(\{i_1,j_1\},\cdots, \{i_{n/2},j_{n/2}\})$
and we introduce for $l=1,\cdots,n/2$ the quantities
$W_{P,l}:=W(v_{i_l},v_{j_l})$ and  $W'_{P,l}:=W'(v_{i_l},v_{j_l})$. 
This yields for the difference of the corresponding correlation functions:
\begin{equation}
\begin{split}
\omega_W(v_1\otimes\cdots\otimes v_n)
&-\omega_{W'}(v_1\otimes\cdots\otimes v_n)
\\
&=\sum_{P \in \Pi_2(n)} \sum_{l=1}^{n/2}
W_{P,1}\cdots W_{P,l-1}(W_{P,l}-W'_{P,l})W'_{P,l+1}\cdots W'_{P,n/2} \; .
\end{split}
\end{equation}
Here we have used that the difference of products can be written as a
sum in the following way:
\begin{equation}
W_{P,1}\cdots W_{P,n/2}-W'_{P,1}\cdots W'_{P,n/2}=\sum_{l=1}^{n/2}
W_{P,1}\cdots W_{P,l-1}(W_{P,l}-W'_{P,l})W'_{P,l+1}\cdots W'_{P,n/2} \; .
\end{equation}
By using the fact that $W$ and $W'$ are bounded bilinear forms, the modulus
of $W_{P,l}$ (and similarly for $W'$ and $W-W'$) can be bounded by
$|W_{P,l}|\leq \|v_{i_l}\|\|v_{j_l}\|\|W\|$ and we obtain the
desired bound 
\begin{equation}
\begin{split}
|\omega_W(v_1\otimes\cdots\otimes v_n)
&-\omega_{W'}(v_1\otimes\cdots\otimes v_n)|
\\
&\leq  \|v_1\|\cdots\|v_n\| \ \|W-W'\| \ |\Pi_2(n)|
\sum_{l=1}^{n/2}\|W\|^{l-1} \|W'\|^{n/2-l} \; .
\end{split}
\end{equation} 
\end{proof}

\section{Proving the bosonic behavior of mean-field fluctuations for large atomic ensembles}
\label{sec:prov-boson-behav}
Roughly, the statement of Proposition~\ref{prop-bosonic} is that mean-field
fluctuations 
for states having weakly $\sqrt{\rm n}$-fluctuations
behave like a bosonic system for large atomic ensembles.    
The following lemma provides the bounds which are used to prove this statement.   

\begin{lem}
\label{lem-00}
Let $\omega_\Lambda$ be a single site homogenous state on the quasi-local
$\6A(\Lambda)$ 
such that the induced net $(\hat\omega_{X})_{X\subset\Lambda}$  has
weakly $\sqrt{\rm n}$-fluctuations. Then for each family of operators
$a_1,\cdots,a_n$ the bound 
\begin{equation}
\begin{split}
|\hat\omega_X(\Phi(a_1)\cdots\Phi(a_{i-1}) I_{\gamma}(a_i,a_{i+1}) \Phi(a_{i+2})\cdots \Phi(a_n))|
\leq 2|X|^{-1/2} \  C_{n-1}\prod_{i=1}^n\|a_i\| 
\end{split}
\end{equation}
holds with $C_n:=\sup_{X\subset\Lambda}\nu_n(\hat\omega_{X})$.  
\end{lem}
\begin{proof}
The fluctuation operators define a representation
$\pi_{\omega_X}(\Phi(a)):=\Phi_{\omega_X}(a)$ of the tensor algebra
$\9T(\6A,*)$ by operators in $\6A(X)$. The induced state $\hat\omega_X=\omega_X\circ\pi_{\omega_X}$ is
just the pullback of the state $\omega_X$ by the representation $\pi_{\omega_X}$.
According to the definition of fluctuation operators, the identity 
\begin{equation}
\pi_{\omega_X}(I_{\gamma}(a,b))=[\Phi_{\omega_X}(a),\Phi_{\omega_X}(b)]-\gamma(a^*,b)\I=
|X|^{-1/2}\Phi_{\omega_X}([a,b])
\end{equation}
is valid for all single site operators $a,b$. Inserting this identity within
the correlation function for $a_1,\cdots, a_n$ implies 
\begin{equation}
\begin{split}
&|\hat\omega_X(\Phi(a_1)\cdots\Phi(a_{i-1}) I_{\gamma}(a_i,a_{i+1}) \Phi(a_{i+2})\cdots
\Phi(a_n))|
\\
&=|X|^{-1/2}|\hat\omega_X(\Phi(a_1)\cdots\Phi(a_{i-1})\Phi([a_i,a_{i+1}])\Phi(a_{i+2})\cdots
\Phi(a_n))|
\\
&\leq 2|X|^{-1/2} \ C_{n-1} \prod_{i=1}^n\|a_i\| \; .
\end{split}
\end{equation}
which proves the proposition. Recall that we have used that all the semi-norms
$\nu_n$ are uniformly bounded.  
\end{proof}

\section{Cluster expansion for  correlation functions (proof of Theorem~\ref{prop-expansion})}
This subsection provides the proof of the expansion
(Theorem~\ref{prop-expansion}) of correlations functions
for states $\omega_\Lambda$ with exponential clustering. We also assume here
that $\omega_\Lambda$ is single site homogenous with single site restriction $\omega$.   

We first derive here an simpler expansion for correlation functions of
the form $\omega_\Lambda(A_1\cdots A_k)=\omega_Y(A_1\cdots A_k)$, 
where $Y=\{y_1,y_2,\cdots,y_k\}\subset \Lambda$ is spread optimally
enumerated and the operator $A_i\in\6A(y_i)$ is localized  at site
$y_i$. The idea of the expansion is to expresses the expectation value $\omega_Y(A_1\cdots A_k)$
in terms of single site expectation values $\omega(A_j)$ as well as
terms $G_j(A_j,A_{j+1}\cdots A_k):=G_{(\{y_j\},\{y_{j+1},\cdots, y_k\})}(A_j,A_{j+1}\cdots A_k)$ that are
given by the bilinear forms $G_{(X,Y)}$. Recall, that for two operators
$A\in\6A(X)$ and $B\in \6A(Y)$ the exponential clustering can be expressed by
the identity
\begin{equation}
\omega_\Lambda(AB)=\omega_X(A)\omega_Y(B)+ G_{(X,Y)}(A,B)\8e^{-d(X,Y)} \; .
\end{equation} 
We use the following lemma in order to prove Proposition~\ref{prop-expansion}:

\begin{lem}
\label{lem-expansion}
Let $Y=\{y_1,\cdots,y_{k}\}$ be spread optimally enumerated and let
$A_i\in\6A(y_i)$, $i=1,\cdots,k$,  be single site operators. Then the
expectation value $\omega_Y(A_1\cdots A_k)$ can be expressed as
\begin{equation}
\begin{split}
\omega_{Y}(A_1\cdots A_k)&=\omega(A_1)\cdots\omega(A_k)
\\
&+\sum_{l=1}^{k-1} \omega(A_1)\cdots\omega(A_{l-1})
G_l(A_l,A_{l+1}\cdots A_k)\8e^{-\Delta(y_l,\cdots, y_k)}
\end{split}
\end{equation}
where the bilinear form $G_l$ is defined as given above.
\end{lem}
\begin{proof}
Let $Y=\{y_1,\cdots,y_{k}\}$ be spread optimally enumerated. Then
$\{y_2,\cdots,y_k\}$ is also spread optimally enumerated. Suppose now the
expansion is valid for spread optimally enumerated sets with $k-1$ elements.  
\begin{equation}
\begin{split}
\omega_{Y\setminus y_1}(A_2\cdots A_k)&=\omega(A_2)\cdots\omega(A_k)
\\
&+\sum_{l=2}^{k-1} \omega(A_2)\cdots\omega(A_{l-1})
G_l(A_l,A_{l+1}\cdots A_k)\8e^{-\Delta(y_l,\cdots, y_k)}
\end{split}
\end{equation}
where the product $\omega(A_2)\cdots\omega(A_{l-1})=1$ is declared to be empty
for $l=2$.
Then we can use the expansion
 \begin{equation}
\begin{split}
\omega_{Y}(A_1A_2\cdots A_k)=\omega(A_1)\omega_{Y\setminus y_1}(A_2\cdots A_k)
+G_1(A_1,A_2\cdots A_k)\8e^{-\Delta(y_1,\cdots, y_k)}
\end{split}
\end{equation}
By inserting the expansion for $\omega_{Y\setminus y_1}(A_2\cdots A_k)$ gives
\begin{equation}
\begin{split}
\omega_{Y}(A_2\cdots A_k)=& \ \ \omega(A_1)\omega(A_2)\cdots\omega(A_k)
\\
&+\sum_{l=2}^{k-1} \omega(A_1)\omega(A_2)\cdots\omega(A_{l-1})
G_l(A_l,A_{l+1}\cdots A_k)\8e^{-\Delta(y_l,\cdots, y_k)}
\\
&+G_1(A_1,A_2\cdots A_k)\8e^{-\Delta(y_1,\cdots, y_k)}
\\
= \ \ &\omega(A_1)\omega(A_2)\cdots\omega(A_k)
\\
&+\sum_{l=1}^{k-1} \omega(A_1)\omega(A_2)\cdots\omega(A_{l-1})
G_l(A_l,A_{l+1}\cdots A_k)\8e^{-\Delta(y_l,\cdots, y_k)}
\end{split}
\end{equation}
Note that for a one-elementary set the statement is trivial and for a two
elementary set $\{y_1,y_2\}$ the expansion is also valid since we have
$\omega_\Lambda(A_1A_2)=\omega(A_1)\omega(A_2)+G_{(y_1,y_2)}(A_1,A_2)\8e^{-d(y_1,y_2)}$
and for any two elementary set the distance coincides with the spread. 
\end{proof}

\begin{proof}[ of Theorem~\ref{prop-expansion}]
Let $\omega_\Lambda$ be a single site homogenous state with strong exponential
clustering. Then for operators $a_1,a_2,\cdots, a_n\in \8{ker}(\omega)$ the
correlation function for the induced state $\hat\omega_X$ can be written as 
\begin{equation}
\begin{split}
\hat\omega_X(a_1\otimes\cdots\otimes a_n)=|X|^{-n/2}\sum_{x\in
  X^n}\omega_X\left(a_1^xa_2^x \cdots  a_{|\8{Ran(x)}|}^x\right) \, .
\end{split}
\end{equation}
where $a_k^x=\prod_{j\in x^{-1}(y_k)} a_j$ are the cluster operators. 
If we enumerate for each $x\in X^n$ the range $\8{Ran}(x)$ spread optimally,
we can apply Lemma~\ref{lem-expansion} to the expectation values
\begin{equation}
\begin{split}
\omega_X\left(a_1^xa_2^x \cdots  a_{|\8{Ran(x)}|}^x\right) 
&=\omega(a_1^x)\cdots \omega(a_{|\8{Ran}(x)|}^x)
\\
&+\sum_{k=1}^{|\8{Ran}(x)|-1} \omega(a_{1}^x)\cdots \omega(a_{k-1}^x)
\\
&\times
G^x_k(a^x_{k},a^x_{k+1}\cdots
a^x_{|\8{Ran}(x)|})\8e^{-\Delta(y_k,y_{k+1},\cdots,y_{|\8{Ran}(x)|})}
\end{split}
\end{equation}
where $G^x_k=G_{(\{y_k\},\{y_{k+1},\cdots,y_{|\8{Ran}(x)|}\})}$. The statement
of  Proposition~\ref{prop-expansion} follows directly by summing over the
elements in $X^n$ and normalizing by $|X|^{-n/2}$.
\end{proof}

\section{Remarks on the strong topology of the tensor algebra}
Let $\omega$ be a state of a C*-algebra $\6A$. Besides the strong topology, we
introduce here the ``$\omega$-strong topology'' on the space of continuous
linear functionals $\9T(\6A,*)^*$.
It is defined to be induced by
the family of semi-norms $\nu_n^\omega$, $n\in\7N$, where $\nu_n^\omega$
assign to each functional $F\in\9T(\6A,*)^*$
the value 
\begin{equation}
\nu_n^\omega(F):=\sup_{a_1,\cdots,a_n\in
  \8{ker}(\omega)}\|a_1\|^{-1}\cdots\|a_n\|^{-1}|F(a_1\otimes\cdots\otimes
a_n)| \; .
\end{equation}
The reason for introducing the semi-norms $\nu_n^\omega$ is that for the
functionals we are dealing with (correlation functions of fluctuation
operators) the semi-norms $\nu_n^\omega$ are easier to estimate. In view of
this, the following lemma is helpful:

\begin{lem}
\label{lem:topol}
For a given state $\omega$ on $\6A$, the $\omega$-strong topology is
equivalent to the strong topology on $\9T(\6A,*)^*$.
In particular the bounds 
\begin{equation}
\nu_n^\omega\leq\nu_n\leq \sum_{k=0}^n {n\choose k} \   2^{k} \ \nu_{k}^\omega
\end{equation}
are valid for all $n\in \7N$. 
\end{lem}
\begin{proof}
Each operator $a\in\6A$ can be written as $a'+c\11$ with
$a'\in\8{ker}(\omega)$ and $c=\omega(a)$. This yields 
\begin{equation}
\begin{split}
a_1\otimes \cdots \otimes a_n&= (a_1'+c_1\11)\otimes \cdots\otimes
(a_n'+c_n\11)
=\sum_{I\subset\{1,2,\cdots n\}} \prod_{j\in I^c}c_j  \
  \bigotimes_{i\in I} a'_i  \; .
\end{split}
\end{equation}
Applying the linear functional on both sides and taking the modulus, we obtain
\begin{equation}
\label{eq:7}
\begin{split}
|F(a_1\otimes \cdots \otimes a_n)|
&\leq \sum_{I\subset\{1,2,\cdots n\}} \prod_{j\in I^c}\|a_j\| \prod_{i\in I}\|a_i'\| \ \nu_{|I|}^\omega(F)
\\
&\leq \|a_1\|\cdots\|a_n\| 
\sum_{I\subset\{1,2,\cdots n\}} 2^{|I|} \ \nu_{|I|}^\omega(F) \, .
\end{split}
\end{equation}
Here we have used that $\|a_i'\|\leq 2\|a_i\|$. This implies the desired bound 
\begin{equation}
\nu_n^\omega\leq \nu_n\leq \sum_{I\subset\{1,2,\cdots n\}} 2^{|I|} \
\nu_{|I|}^\omega \; .
\end{equation}
Note that the semi-norm $\nu_n^\omega$ is optimizing the modulus of the value
of a functional over tuples $(a_1,\cdots,a_n)$ in the kernel $\8{ker}(\omega)$. This immediately
implies that $\nu_n\geq\nu_n^\omega$. 
\end{proof}

\section{States with exponential
  clustering have weakly $\sqrt{\8n}$-fluctuations (proof of Theorem~\ref{thm-main})}
\label{sec:stat-with-expon}

\subsection{Homogenous product states (proof of Proposition~\ref{prop-prod})}
\begin{proof}[ of Proposition~\ref{prop-prod}]
Let $\omega_\Lambda=\omega^{\otimes\Lambda}$ be a homogenous product
state. For operators $a_1,\cdots,a_n\in\8{ker}(\omega)$ the correlation
function of the induced state $\hat\omega_X$ is given by
\begin{equation}
\label{eq:5}
\begin{split}
\hat\omega^{\otimes X}(a_1\otimes\cdots\otimes a_n)
&=|X|^{-\frac{n}{2}}\sum_{x\in X^n}\omega(a_{1}^x)\cdots\omega(a_{|\8{Ran}(x)|}^x)
 \; ,
\end{split}
\end{equation}
where $a_k^x$ is the cluster operator $a_k^x=\prod_{j\in x^{-1}(y_k)}
a_j$. Let $E_k(X)$ the set of all enumerated $k$-elementary subsets in $X$ and
let $\Pi(k,n)$ the set of all ordered partitions of $\{1,\cdots,n\}$ into $k$
non-empty subsets. Then (\ref{eq:5}) can be written as 
 \begin{equation}
\label{eq:5}
\begin{split}
\hat\omega^{\otimes X}(a_1\otimes\cdots\otimes a_n)
&=|X|^{-\frac{n}{2}}\sum_{k=1}^n\sum_{\{y_1,\cdots,y_k\}\in
  E_k(X)}\sum_{(I_1,\cdots,I_k)\in \Pi(k,n)}\omega(a_{I_1})\cdots\omega(a_{I_k})
\\
&=|X|^{-\frac{n}{2}}\sum_{k=1}^n|E_k(X)|\sum_{(I_1,\cdots,I_k)\in \Pi(k,n)}\omega(a_{I_1})\cdots\omega(a_{I_k})
 \; ,
\end{split}
\end{equation}
with cluster operators $a_{I_j}:=\prod_{i\in I_j} a_i$. Here we have used the
fact that for each tuple $x\in X^n$ there is a unique enumerated
$k$-elementary subset $\{y_1,\cdots, y_k\}$ ($k\leq n$) and a partition $(I_1,\cdots,I_k)\in\Pi(k,n)$
such that $x_i=y_j$ for $i\in I_j$.
Suppose that for a partition $(I_1,\cdots,I_k)$ one of the sets $I_j$ contains
only one element $I_j=\{l\}$, then this partition is not contributing to the
sum since $\omega(a_{I_j})=\omega(a_l)=0$. Hence, only those partitions with
$|I_j|\geq 2$ are contributing. For the $n$-point correlation function
(restricted to $\8{ker}(\omega)^{\otimes n}$) we
can write 
\begin{equation}
\label{eq:6}
\begin{split}
\hat\omega^{\otimes X}(a_1\otimes\cdots\otimes a_{n})
&=\prod_{l=0}^{n/2-1}\left(1-\frac{l}{|X|}\right)\sum_{(I_1,\cdots,I_{n/2})\in \Pi_2(n)}\omega(a_{I_1})\cdots\omega(a_{I_{n/2}})
\\
&+\sum_{k=1,k\not=n/2}^{n}|X|^{-n/2+k}\prod_{l=0}^{k-1}\left(1-\frac{l}{|X|}\right)
\sum_{(I_1,\cdots,I_k)\in\Pi_{>2}(k,n)}\omega(a_{I_1})\cdots\omega(a_{I_k}) 
\; ,
\end{split}
\end{equation}
where $\Pi_2(n)$ is the set of all ordered partitions into two elementary
subsets and $\Pi_{>2}(k,n)$ is the set of all ordered partitions into $k$
subsets containing more than one element $(I_1,\cdots, I_k)$, where at least one subset $I_j$
contains more than two elements. Note, that the first term in (\ref{eq:6}) is
vanishing if $n$ is odd. $E_k(X)$ is the set of enumerated $k$-elementary
subsets whose cardinality is 
$|E_{k}(X)|=\prod_{l=0}^{k-1}(|X|-l)$. Restricted to the sub-algebra
$\9T(\8{ker}(\omega),*)$, the functional $\hat\omega^{\otimes X}$ can be
written as  
\begin{equation}
\label{eq:8}
\hat\omega^{\otimes
  X}=\prod_{l=0}^{n/2-1}\left(1-\frac{l}{|X|}\right)\hat\omega_{\8{qf}} +D_X
\end{equation}
with a continuous functional $D_X\in\9T(\6A,*)^*$ that is given by 
\begin{equation}
\begin{split}
D_X(a_1\otimes \cdots\otimes a_n)
&:=\sum_{k=1}^{n}|X|^{-n/2+k}\prod_{l=0}^{k-1}\left(1-\frac{l}{|X|}\right)
\sum_{(I_1,\cdots,I_k)\in\Pi_{>2}(k,n)}\omega(a_{I_1})\cdots\omega(a_{I_k}) 
\end{split}
\end{equation}
for $a_1,\cdots, a_n\in \8{ker}(\omega)$. To bound the semi-norm
$\nu_n^\omega(D_X)$, we observe the bound  
\begin{equation}
\begin{split}
|D_X(a_1\otimes \cdots\otimes a_n)|
&\leq \|a_1\|\cdots \|a_n\| \ \sum_{k=1}^{n}|X|^{-n/2+k}\prod_{l=0}^{k-1}\left(1-\frac{l}{|X|}\right) |\Pi_{>2}(k,n)| 
\; .
\end{split}
\end{equation}
If we consider a partition $(I_1,\cdots,I_k)$ in $\Pi_{>2}(k,n)$ then the
constraint $2k+1\leq n$ has to be fulfilled. This implies $|X|^{-n/2+k}\leq
|X|^{-1/2}$ and we obtain:
\begin{equation}
\label{eq:4}
\begin{split}
|D_X(a_1\otimes \cdots\otimes a_n)|
&\leq \|a_1\|\cdots \|a_n\| \
|X|^{-1/2}\sum_{k=1}^{n}\prod_{l=0}^{k-1}\left(1-\frac{l}{|X|}\right) S(k,n)
\delta(2k+1\leq n)
\; ,
\end{split}
\end{equation}
where $S(k,n)=|\Pi(k,n)|$ is the number of all partitions of the set $\{1,\cdots,n\}$
into $k$ non-empty subsets (Stirling number of second kind). For a logical statement
$S$, the $\delta$-function is defined by $\delta(S)=1$, if $S$ is true, and
$\delta(S)=0$ if $S$ is false. By
Lemma~\ref{lem:topol} it is sufficient to show convergence for the semi-norms
$\nu_n^\omega$. From the inequality (\ref{eq:4}) we obtain for $|X|>n$ the bound
\begin{equation}
\begin{split}
\nu_n^\omega(D_X)
&\leq \
|X|^{-1/2}\sum_{k=1}^{n}S(k,n)
\delta(2k+1\leq n)
\end{split}
\end{equation} 
which implies that $s-\lim_{X\subset \Lambda} D_X=0$. Since
$\lim_{X\subset\Lambda}\prod_{l=0}^{n/2-1}\left(1-\frac{l}{|X|}\right)=1$, the
proposition follows finally from (\ref{eq:8}): 
\begin{equation}
s-\lim_{X\subset\Lambda}\hat\omega^{\otimes
  X}=s-\lim_{X\subset\Lambda}\prod_{l=0}^{n/2-1}\left(1-\frac{l}{|X|}\right)\hat\omega_{\8{qf}}
=\left(\lim_{X\subset\Lambda}\prod_{l=0}^{n/2-1}\left(1-\frac{l}{|X|}\right)\right)
\hat\omega_{\8{qf}}=\hat\omega_{\8{qf}} \; .
\end{equation}

\end{proof}

\subsection{Estimating the number of  subsets for a given spread}
To formulate our next lemma, we introduce the function $\Delta$ that assigns
to each finite subset $Y\subset\Lambda$ the maximal distance of a point in $Y$
to its complement: $\Delta(Y):=\max_{y\in Y} d(y,Y\setminus y)$. To prove
bound on correlation functions, one partial task is to count for a finite
subset $X\subset\Lambda$ the number 
$N(X,k,r)$ of all $k$-elementary ($k\geq 2$) subsets $Y$ in $X$ such
that $\Delta(Y)\leq r$. Whereas the spread $\Delta$ measures the spreading of
one point sets, we can also look at the spreading of sets that contain more
elements. We introduce the $k$-spread of a set $Y$ as
$\Delta_k(Y)=\max_{J\in P_k(Y)}d(J,Y\setminus J)$ which measures the
largest distance of a $k$-elementary subset $J\in P_k(Y)$ to its relative
complement in $Y$ (here $P_k(Y)$
denotes the set of all $k$-elementary subsets in $Y$). According to this
definition, the one-spread $\Delta_1=\Delta$ is just the spread.   

\begin{lem}
\label{lemma-tech-0}
Let $Y$ be a subset of $X$ with $\Delta(Y)\leq r$ and $\Delta_2(Y)>r$. Then
$Y=Z\cup \{x,y\}$ is the disjoint union of a set $Z$ and a two elementary
subset $\{x,y\}$ such that $\Delta(Z)\leq r$ and $d(x,y)\leq r$.
\end{lem} 
\begin{proof}
Since the 2-spread of $Y$ is larger than $r$, there exists a pair of points $\{x,y\}$ such that
$d(\{x,y\},Y\setminus\{x,y\})>r$. Let $Z=Y\setminus\{x,y\}$ be the relative
complement of $\{x,y\}$. Then the spread of $Y$ can be expressed as 
\begin{equation}
\Delta(Y)=\max_{z\in Z}\left\{\min\{d(z,Z\setminus
  z),d(z,x),d(z,y)\},\min\{d(x,Z),d(x,y)\},\min\{d(y,Z),d(x,y)\}\right\} \; .
\end{equation}
Since $\Delta(Y)\leq r$ the bound $\min\{d(y,Z),d(x,y)\}\leq r$ has to be
fulfilled. This implies that $d(x,y)\leq r$ because  $d(y,Z)>r$ holds by our
assumption. Moreover, the inequality 
\begin{equation}
\Delta(Y)\geq \max_{z\in Z}\left\{\min\{d(z,Z\setminus
  z),d(z,x),d(z,y)\}\right\}
\end{equation}
is obviously fulfilled. For each $z\in Z$ we have $r\geq \min\{d(z,Z\setminus
z),d(z,x),d(z,y)\}=d(z,Z\setminus z)$ 
since $d(z,x)>r$ and $d(z,y)>r$ is valid according to our assumption that each point in $Z$ has a distance $>r$
  to $x$ and $y$. This implies the inequality
\begin{equation}
r\geq \Delta(Y)\geq \max_{z\in Z}d(z,Z\setminus
  z) =\Delta(Z) \; .
\end{equation}  
Hence, $Y$ can be decomposed into a disjoint union of a set $Z$ and two points
$\{x,y\}$ such that $\Delta(Z)\leq r$ and $d(x,y)\leq r$.
\end{proof}

\begin{lem}
\label{lemma-tech-1}
Let $Y$ be a subset such that $\Delta(Y)\leq r$ and $\Delta_2(Y)\leq r$. Then
$Y=Z\cup \{x\}$ is the disjoint union of a set $Z$ and a single point $x$ such
that $\Delta(Z)\leq r$. 
\end{lem}
\begin{proof}
Let $\Delta_2(Y)\leq r$, i.e. for each pair $\{x,y\}$, the distance to the relative complement
$d(\{x,y\},Y\setminus\{x,y\})\leq r$. In order to discuss this case, we
introduce for each point $y\in Y$ the number $I(y,Y,r)=|\{z\in Y\setminus y|d(y,z)\leq
r\}|$ of points in the relative complement of $y$ whose distance to $y$ is
smaller than $r$. Now we have to perform a further case distinction:
\begin{itemize}
\item
There exists a point $x\in Y$ such that $I(x,Y,r)=1$. In this case the spread
of $Y$ can be expressed in terms of the set $Z=Y\setminus x$ and the single
point $x$ as
\begin{equation}
\begin{split}
\Delta(Y)=&\max_{z\in Z}\left\{\min\{d(z,Z\setminus
  z),d(z,x)\},d(x,Z) \right\} \; .
\end{split}
\end{equation}
Suppose now that $d(x,z)\leq r$ holds. Then this is true for a unique point $z=z^*\in
Z$. Since the 2-spread fulfills the bound $\Delta_2(Y)\leq r$, for all 
$y\in Z\setminus z^*$ the bound $d(\{x,z^*\},y)\leq r$ follows. This implies
that $d(z^*,y)\leq r$ since  $d(x,y)>r$ holds for all points $y\in Z\setminus z^*$. 
Hence $d(z^*,Z\setminus z^*)\leq r$. In all other cases, we have $d(x,z)>
r$ which implies $r\geq \min\{d(z,Z\setminus
  z),d(z,x)\}=d(z,Z\setminus z)$. 
Putting these things together implies finally 
\begin{equation}
\begin{split}
\Delta(Z)=\max_{z\in Z}d(z,Z\setminus z)\leq r \; .
\end{split}
\end{equation}

\item
For all points $y\in Y$ we have  $I(y,Y,r)\geq 2$. In this case we can take
any point $x$ and express the spread of $Y$ as
\begin{equation}
\begin{split}
\Delta(Y)=&\max_{z\in Z}\left\{\min\{d(z,Z\setminus
z),d(z,x)\},d(x,Z)\right\} \; .
\end{split}
\end{equation}
with $Z=Y\setminus x$. Suppose now that $d(z,x)\leq r$ for some $z\in
Z$. Since  $I(x,Y,r)\geq 2$ there exists at least one $y\in Z$ such that
$d(z,y)\leq r$ which implies $d(z,Z\setminus
z)\leq r$. On the other hand, if $d(z,x)>r$, then $d(z,Z\setminus
z)\leq r$ since $\min\{d(z,Z\setminus
z),d(z,x)\}\leq r$. This implies again
\begin{equation}
\begin{split}
\Delta(Z)=\max_{z\in Z}d(z,Z\setminus z)\leq r \; .
\end{split}
\end{equation}
\end{itemize}
\end{proof}

\begin{lem}
\label{lem-01}
For each finite subset $X$, for each $r>0$, and  for each $n\geq 2$ the
bounds 
\begin{equation}
\begin{split}
N(X,n,r)\leq q_{[n]} |X|^{n/2} N(r)^{n/2}
\end{split}
\end{equation} 
are valid, where the numbers $q_{[n]}$ are recursively determined by 
$q_{[n+2]}=(n+1)q_{[n+1]}+q_{[n]}$ with initial conditions $q_{[2]}=1$ and $q_{[3]}=2$.
\end{lem}
\begin{proof}
According to Lemma~\ref{lemma-tech-0} and Lemma~\ref{lemma-tech-1}, we have the
following two cases:
\begin{enumerate}
\item 
Each subset $Y\subset X$ with $\Delta(Y)\leq r$ is a disjoint union $Y=Z\cup
\{x,y\}$ with $\Delta(Z)\leq r$ and $d(x,y)\leq r$.
\item 
$Y$ is a disjoint union $Y=Z\cup \{x\}$ with $\Delta(Z)\leq r$. 
\end{enumerate}
Putting both cases together, we obtain a recursion formula for bounding
$N(X,n+2,r)$. For the first case, we can choose for each $n$-elementary subset
$Z$ with $\Delta(Z)\leq r$ a  pair $\{x,y\}$ with $d(x,y)\leq r$. This gives
less than $N(X,n,r)|X|N(r)$ possibilities. We have to add the second case,
where for each $n+1$-elementary subset $Z$ with $\Delta(Z)\leq r$ we can just
add a point $x$ such that $\Delta(Z\cup\{x\})\leq r$. Since there must be a
point $z\in Z$ with $d(x,z)\leq r$, we have $(n+1)N(r)$ possibilities to
choose $x$ which gives less than $(n+1)N(r)N(X,n+1,r)$ possibilities. In total,  
we obtain the recursive bound 
\begin{equation}
N(X,n+2,r)\leq (n+1)N(r)N(X,n+1,r)+|X|N(r)N(X,n,r)  \; .
\end{equation}
To get a convenient explicit solution of this recursion, we allow to
over-count a bit by using the inequality  $N(r)\leq N(r)^{1/2}|X|^{1/2}$:  
\begin{equation}
N(X,n+2,r)\leq (n+1)N(r)^{1/2}|X|^{1/2} N(X,n+1,r)+|X|N(r)N(X,n,r)  \; .
\end{equation}
Now, we insert the ansatz $N(X,n,r)\leq q_{[n]} |X|^{n/2}N(r)^{n/2}$ into the
recursion bound which gives the following consistency relation:
\begin{equation}
\begin{split}
q_{[n+2]} |X|^{(n+2)/2}N(r)^{(n+2)/2} 
&\leq 
(n+1)N(r)^{1/2}|X|^{1/2}q_{[n+1]}
|X|^{(n+1)/2}N(r)^{(n+1)/2} 
\\
&+|X|N(r)q_{[n]} |X|^{n/2}N(r)^{n/2}
\\
&\leq 
\left[(n+1)q_{[n+1]}+q_{[n]}\right] |X|^{(n+2)/2}N(r)^{(n+2)/2}
\end{split}
\end{equation}
Thus, if the sequence $q_{[n]}$ fulfills the recursion relation  
$q_{[n+2]}=(n+1)q_{[n+1]}+q_{[n]}$, then we obtain a consistent upper bound.
For getting the correct 
initial conditions with respect to $n$, we first look
at the case $k=2$. The first point of the requested set can be chosen 
freely in $X$, which gives $|X|$ possibilities. The second point must have
distance $r$ to the one firstly chosen. Hence we obtain the bound $N(X,2,r)\leq |X|N(r)$. 
For $k=3$ two points within the requested set must have distance $r$. This
gives at most $|X|N(r)$ possibilities. The third point must have a distance
$\leq r$ to one of the chosen two points. This yields  $N(X,3,r)\leq
2|X|N(r)^2\leq 2|X|^{3/2}N(r)^{3/2}$.
\end{proof}

\subsection{Towards bounding  the semi-norms $\nu_n$}
Recall that $\Lambda$ is a countably infinite lattice with a {\em regular}
distance $d$. Here regular means that the maximal number of lattice sites
within a ball of radius $r$ is bounded by a polynomial $P(r)$, i.e. 
$N(r)=\sup_{x\in\Lambda}|\{y\in\Lambda|d(x,y)\leq r\}|\leq P(r)$ for all $r\in\7R_+$. 

If we require that the distance $d$ is regular, than $N(r)$ is bounded by a
polynomial which implies that $\hat N_k=\sum_{r=1}^\infty N(r)^k e^{-r}$ is
finite and monotonous increasing in $k$. For each finite subset $X$, the
semi-norms $\nu_n$ of the induced state
$\hat\omega_X=\hat\omega^{\otimes X}+F_X$ can be bounded by bounding the
corresponding semi-norms of $\hat\omega^{\otimes X}$ and $F_X$ separately.   
We know already from Proposition~\ref{prop-prod} that the semi-norms of 
$\hat\omega^{\otimes X}$ are uniformly bounded in $X$. Thus it is sufficient
to bound the semi-norms of $F_X$. For this purpose, we introduce the
quantities    
\begin{equation}
\begin{split}
B_n(X):=\sum_{x\in X^n}
\sum_{k=1}^{|\8{Ran}(x)|-1}
\prod_{l=1}^{k-1}\delta(|x^{-1}(y_l)|>1)\8e^{-\Delta(y_k,y_{k+1},\cdots,y_{|\8{Ran}(x)|})}
\; .
\end{split}
\end{equation}

\begin{lem}
\label{lem-02}
For each finite subset $X\subset\Lambda$, $B_n(X)$ fulfill the bound
\begin{equation}
\begin{split}
B_n(X)\leq  \hat B_n   \ |X|^{n/2}
\end{split}
\end{equation}
for each finite subset $X\subset\Lambda$,
where $\hat B_n$ are finite positive numbers that are given by 
\begin{equation}
\begin{split}
\hat B_n=  n!\sum_{k=1}^{n} S(k,n)\sum_{l=1}^{k-1}  q_{[k-l+1]} \hat N_{k-l+1}  \; .
\end{split}
\end{equation}
Here $S(k,n)$ are the Stirling
numbers of the second kind and the numbers $q_{[n]}$ are recursively determined by 
$q_{[n+2]}=(n+1)q_{[n+1]}+q_{[n]}$ with initial conditions $q_{[2]}=1$ and $q_{[3]}=2$. 
\end{lem}  
\begin{proof}
To bound $B_n(X)$, we take advantage of the following fact: Let $P(n,k)$ be
the set of partitions of $n$ into $k$ non-vanishing summands. For each tuple
$x\in X^n$ there exists a 
$k$-elementary subset $Y=\{y_1,\cdots,y_k\}\subset X$
with $k\leq n$ and a partition $(n_1,\cdots,n_k)\in P(k,n)$ as
well as a permutation $\sigma\in S_n$ such that $x=\sigma(y_1^{\times
  n_1},\cdots, y_k^{\times n_k})$ where $y^{\times l}$ is the $l$-tuple with
constant entry $y$.  Note that the permutation $\sigma$ is not unique, since
permuting one of the sub-tuples $y_j^{\times n_j}$ leaves $x$ invariant. 
This implies the following bound for $B_n(X)$:
\begin{equation}
\begin{split}
B_n(X)\leq n!\sum_{k=1}^{n} \sum_{\{y_1,\cdots ,y_k\}\subset X}
\sum_{(n_1,\cdots,n_k)\in P(k,n)}
\sum_{l=1}^{k-1}\delta(n_1>1)\cdots\delta(n_{l-1}>1)\8e^{-\Delta(y_{l},\cdots,y_{k})}
\; .
\end{split}
\end{equation}
For fixed $l<k$, each partition $(n_1,\cdots, n_k)$ which contributes to the sum has to fulfill
the condition $l+k-1\leq n$. The constrained given by the ``delta-functions'' $\delta(n_1>1)\cdots\delta(n_{l-1}>1)$
can be replaced by the weaker constraint $\delta(l+k-1\leq n)$  which yields
an upper bound of the right hand side: 
\begin{equation}
\begin{split}
B_n(X)\leq n!\sum_{k=1}^{n} 
S(k,n)\sum_{l=1}^{k-1}\delta(l+k-1\leq n)\sum_{\{y_1,\cdots ,y_k\}\subset X} \8e^{-\Delta(y_{l},\cdots,y_{k})} 
\end{split}
\end{equation}
For each $l<k$ the sum over $k$-elementary subsets can be estimated as
follows:
\begin{equation}
\begin{split}
\sum_{\{y_1,\cdots ,y_k\}\subset X} \8e^{-\Delta(y_{l},\cdots,y_{k})} &\leq
|X|^{l-1}\sum_{\{y_{l},\cdots,y_{k}\}\subset
  X}e^{-\Delta(\{y_{l},\cdots,y_{k}\})} 
\\
&\leq |X|^{l-1}\sum_{r=1}^{\infty}N(X,k-l+1,r) \8e^{-r}
\\
&\leq q_{[k-l+1]}|X|^{(k+l-1)/2} \hat N_{k-l+1} 
\end{split}
\end{equation}
where we have used Lemma~\ref{lem-01}. This yields the desired bound  
\begin{equation}
\begin{split}
B_n(X)\leq n!\sum_{k=1}^{n} 
S(k,n)\sum_{l=1}^{k-1}  q_{[k-l+1]} \hat N_{k-l+1} \ |X|^{n/2}
\end{split}
\end{equation}
which completes the proof of the lemma.
\end{proof}

\begin{proof}[ of Theorem~\ref{thm-main}]
By taking advantage of the expansion of Proposition~\ref{prop-expansion} we
can write the state $\hat\omega_X=\hat\omega^{\otimes X}+F_X$. Let
us assume first, that the operators 
$a_1,\cdots,a_n$ belong to the kernel of the single site restriction
$\omega$. We first bound the correlation function of the functional $F_X$. 
\begin{equation}
\begin{split}
|F_X&(a_1\otimes\cdots\otimes a_n)|
\\
&\leq |X|^{-n/2}\sum_{x\in
  X^n}\sum_{k=1}^{|\8{Ran}(x)|-1} \prod_{l=1}^{k-1}|\omega(a_{y_l}^x)|
|G^x_k(a^x_{y_k},a^x_{y_{k+1}}\cdots
a^x_{y_{|\8{Ran}(x)|}})|\8e^{-\Delta(y_k,y_{k+1},\cdots,y_{|\8{Ran}(x)|})}
\\
&\leq |X|^{-n/2}\|a_1\|\cdots\|a_n\| G_0\sum_{x\in
  X^n}\sum_{k=1}^{|\8{Ran}(x)|-1} \prod_{l=1}^{k-1}\delta(|x^{-1}(y_l)|>1)
\8e^{-\Delta(y_k,y_{k+1},\cdots,y_{|\8{Ran}(x)|})}
 \; ,
\end{split}
\end{equation}
where we have used the fact that $\omega(a_y^x)=0$ whenever the pre-image
$|x^{-1}(y)|= 1$ contains only one element. By inserting the definition of the
quantities $B_n(X)$ we find from Lemma~\ref{lem-02}
\begin{equation}
\nu_n^\omega(F_X)\leq |X|^{-n/2} G_0 B_n(X) \; .
\end{equation}
By Proposition~\ref{prop-prod} the semi-norms
$\nu_n^\omega(\hat\omega^{\otimes X})$ 
are uniformly bounded, i.e. for each
$n\in\7N$, the constant $\hat
A_n:=\sup_{X\subset\Lambda}\nu_{n}^\omega(\hat\omega^{\otimes X})<\infty$ is
finite. Therefore we obtain for the induced state $\hat\omega_X$ the semi-norm
bound:
\begin{equation}
\begin{split}
\nu_{n}^\omega(\hat\omega_X)\leq \hat A_{n}+G_0\hat B_{n}
\end{split}
\end{equation}
Hence the semi-norm $\nu_n$ can be bounded by $\nu_n^\omega$ according to Lemma~\ref{lem:topol}: 
\begin{equation}
\begin{split}
\nu_n(F)\leq 
\sum_{k=0}^n {n\choose k} 2^{k} \ \nu_{k}^\omega(F) \; .
\end{split}
\end{equation}
By applying this bound to $\hat\omega_X$ we obtain   
\begin{equation}
\begin{split}
\nu_n(\hat\omega_X)
&\leq 
\sum_{k=0}^n 2^{k}  \ {n\choose k} \ (\hat A_{k}+G_0\hat B_{k}) \; .
\end{split}
\end{equation}
This implies that $\sup_{X\subset\Lambda}\nu_n(\hat\omega_X)<\infty$ and the
net $(\hat\omega_X)_{X\subset \Lambda}$ has weakly $\sqrt{\rm n}$-fluctuations. 
\end{proof}

\end{appendix}

\end{document}